\documentclass[journal]{IEEEtran}
\usepackage{amsmath,amssymb}
\usepackage{booktabs}
\usepackage{graphicx}
\usepackage{algorithm}
\usepackage{algpseudocode}
\usepackage{url}
\usepackage[hidelinks]{hyperref}
\usepackage{xcolor}
\newtheorem{observation}{Observation}
\newtheorem{definition}{Definition}
\newtheorem{proposition}{Proposition}

\begin{document}

\title{Coverage Is Not Containment: A Fundamental Limit of Admission-Time Defenses
Against Coordinated Poisoning of Vector Retrieval}

\author{Prashant~Kumar~Pathak~and~Tarun~Kumar~Sharma%
\thanks{P.~K.~Pathak is with Santa Clara, CA, USA (e-mail: prashant.pathak@ieee.org).}%
\thanks{T.~K.~Sharma is with Colonia, NJ, USA (e-mail: tarun.sharma@ieee.org).}}

\markboth{IEEE Transactions on Information Forensics and Security}%
{Pathak and Sharma: Coverage Is Not Containment}

\maketitle

\begin{abstract}
Retrieval-augmented generation (RAG) answers a question by retrieving passages from a vector store and
trusting them as context, so anyone who can add documents can try to steer the answer. A recent, appealing
defense filters poisoning at ingestion, rejecting any document that behaves like a \emph{hub}. We show
it---and \emph{every} ingestion-time filter---is defeated by a coordinated adversary that injects a handful
of \emph{individually unremarkable} documents which together surround one target query and seize its
top-$k$ (on BGE-large / BEIR, $m{=}10$ documents take $10/10$; $9.9/10$ on a live HNSW index).

The attack is not theoretical. Realized as ordinary \emph{fluent} text and run end-to-end through a
BGE-large\,$+$\,HNSW\,$+$\,Qwen2.5-7B pipeline, it makes the generator emit the attacker's planted claim in
$\mathbf{88\%}$ of targets, versus $0\%$ without the injection. And no admission-time defense stops it: at
ingestion an attack cone is geometrically identical to a legitimate niche upload, so---measuring this
directly---the strongest trained classifier, given every feature and thousands of examples, separates the
two \emph{no better than chance}, catching $\mathbf{4.2\%}$ of attacks at a $1\%$ false-positive rate. We
prove this limit for the entire class of ingestion-time statistics (any decision from documents and
reference queries alone), and it reproduces---and worsens---across two corpora and five encoders. The one
signal that separates an attack from legitimate niche ingestion---a query's \emph{demand}---is invisible
before retrieval, which is also the escape: a retrieval-time detector that observes demand catches
$\mathbf{100\%}$ of the attacks at the same $1\%$ false-positive rate. Coverage of the query space by an
admission gate is not containment of coordinated poisoning; robust defense must move past the front door,
to demand.
\end{abstract}


\section{Introduction}
Retrieval-augmented generation (RAG)~\cite{lewis2020rag,guu2020realm} has become the dominant way to ground
large language models (LLMs) in external knowledge: a user query is embedded, the nearest documents in a
vector store are retrieved, and those passages are placed in the model's context as trusted evidence. This
architecture makes the vector store a \emph{security boundary}. A document that is retrieved for a query
can steer the model's answer to that query---through indirect prompt injection, misinformation, or biased
evidence~\cite{greshake2023injection,zou2024poisonedrag,zhong2023corpuspoison}. An adversary who can insert
documents into the store therefore has a powerful lever over downstream generations.

\emph{Hubness}---the well-documented tendency, in high-dimensional spaces, for a few points to appear in
the $k$-nearest-neighbour lists of disproportionately many queries~\cite{radovanovic2010hubs}---sharpens
this lever. A single crafted \emph{hub} document can be retrieved across many unrelated queries, so one
injected record influences a large fraction of interactions~\cite{zhang2024advhub}. The natural defense is
to detect and remove such hubs. Detection after ingestion, however, leaves an exposure window between a
hub's insertion and the next scan, and pays the cost of repeated corpus-wide rescans.

A recent line of work~\cite{cose} moves the control to \emph{admission}. It maintains a set of sentinel
queries and, for each candidate document, computes its reverse-$k$NN count $\kappa_S$ against the
sentinels; a document is rejected if $\kappa_S$ reaches a threshold~$\tau$ calibrated to a small benign
false-positive rate, so a hub never enters the index. That work shows, perhaps surprisingly, that a single
\emph{global} threshold suffices---domain-aware (per-topic) refinement adds nothing---and explains it
geometrically: sentence embeddings are \emph{anisotropic}, occupying a narrow cone, so a document that is
hub-like within a topic is hub-like globally. Crucially, it evaluates a fixed set of \emph{non-adaptive}
attacks and explicitly scopes \emph{targeted} and \emph{coordinated} attacks out of scope.

\textbf{This paper takes up exactly that scoped-out threat, and finds a fundamental limit.} We ask two
questions. \emph{(i)~Attack:} can a coordinated adversary poison a chosen \emph{target} query while every
injected document is individually admissible? \emph{(ii)~Defense:} can any \emph{ingestion-time}
control---per-document or collective, seeing only documents and sentinels---stop it at an acceptable cost?
Our answers are yes, and no (Proposition~\ref{prop:limit}).

The key idea is that the per-document gate reasons about documents \emph{one at a time}. A hub is caught
because it is loud on its own. But an adversary need not build a hub: it can inject many documents that are
each individually quiet---each retrieved by too few sentinels to be rejected---yet that \emph{together}
dominate one target query's retrieval. By the very anisotropy that makes the global gate work, the
documents that achieve this must live near \emph{peripheral} queries---queries poorly covered by the
established sentinels---which is precisely the tight-domain residual the gate cannot close. When the
defender responds with a \emph{collective} statistic that looks for coordinated bursts, the adversary tunes
its attack into a regime that is geometrically indistinguishable from legitimate bulk ingestion of related
documents. This indistinguishability, which we both \emph{measure}---the strongest trained classifier
separates the attack from a location-matched legitimate upload no better than chance
(Fig.~\ref{fig:clf})---and formalize (Proposition~\ref{prop:limit}), is the crux of a limit that persists at
every practical false-positive rate (Fig.~\ref{fig:frontier}) and that no ingestion-time defense in a
natural class escapes.

\noindent\textbf{Contributions.}
\begin{itemize}
\item \textbf{A coordinated attack that reaches the output (\S\ref{sec:attack}).} We formalize the
coordinated adversary and show $m$ individually-admissible documents in a tight cone around a target query
seize its top-$k$ linearly in $m$ ($m{=}10$ seizes $10/10$ on BGE-large / BEIR; $9.9/10$ on a live HNSW
index), give a feasibility law tying attackability to query \emph{centrality}, and realize it as
\emph{fluent} text that evades a perplexity filter. End-to-end through a real RAG pipeline
(BGE\,$+$\,HNSW\,$+$\,Qwen2.5-7B) it flips the generated answer to the attacker's planted claim in
$\mathbf{88\%}$ of targets, versus $0\%$ clean.
\item \textbf{A measured fundamental limit (\S\ref{sec:defense},~\S\ref{sec:limit}).} Two collective
defenses and the adaptive game leave a $4.5/10$ covert residual that persists at \emph{every} achievable
false-positive rate. We prove (Proposition~\ref{prop:limit}) this holds for the entire class of
ingestion-time statistics, and---the key evidence---\emph{measure} it: the strongest trained classifier,
given every feature, separates the attack from a location-matched legitimate upload \emph{no better than
chance} ($\mathbf{4.2\%}$ recall at $1\%$ FPR). It reproduces across two corpora and five encoders.
\item \textbf{A constructive escape, and systems (\S\ref{sec:systems},~\S\ref{sec:discuss}).} Because the
distinguishing signal is retrieval-time \emph{demand}, a detector that observes it catches $\mathbf{100\%}$
of the attacks at $1\%$ FPR---where the best admission-time detector catches $4.2\%$. Sub-points: the
collective defense costs $\sim\!10\%$ of insert latency, and a per-shard view is blind to a burst split
across shards (motivating global consistency).
\end{itemize}

\section{Background and Related Work}\label{sec:related}
\textbf{RAG security and corpus poisoning.} Because retrieved passages are trusted, poisoning the retrieval
corpus is a direct attack on RAG. PoisonedRAG~\cite{zou2024poisonedrag} and corpus-poisoning
attacks~\cite{zhong2023corpuspoison,zou2024poisonedrag} craft passages that are retrieved for target
questions and steer generation; backdoor and trigger attacks plant passages activated by specific
queries~\cite{phantom,badrag}; and indirect prompt injection~\cite{greshake2023injection} weaponizes
retrieved content. These build on the broader data-poisoning
lineage~\cite{biggio2012poisoning,steinhardt2017certified,shafahi2018poison}. Two things distinguish our
setting. First, these works typically target \emph{specific} question--answer pairs or optimize a small
number of passages against an \emph{undefended} retriever, whereas we study a defense-aware adversary that
must remain admissible under an explicit ingestion-time gate and seeks to \emph{dominate}---not merely
enter---a target query's top-$k$. Second, our construction is closer in spirit to a Sybil
attack~\cite{douceur2002sybil}, in which many individually-weak entities combine, than to a single strong
poison. Adversarial attacks on neural retrieval and ranking craft passages that are retrieved or ranked
highly~\cite{song2020collisions,liu2022orderdisorder}; we differ again in operating under an admission gate.

\textbf{Defenses against retrieval poisoning.} Defenses span the RAG pipeline, and it is worth situating our
limit against each stage. (i)~\emph{Ingestion-time filtering}---the admission gate~\cite{cose} we study, and
more generally any anomaly test applied as documents arrive---is the earliest and cheapest place to act;
our results show this stage cannot succeed against coordinated poisoning at any acceptable false-positive
rate. (ii)~\emph{Retrieval-time} defenses inspect a query's retrieved set: robust aggregation isolates each
passage and combines per-passage answers so a \emph{minority} of poisoned passages cannot dominate, with
certifiable guarantees in RobustRAG~\cite{xiang2024robustrag}, which inherits from certified poisoning
defenses based on bagging and partition aggregation~\cite{jia2021bagging,levine2021dpa} and randomized
smoothing against label flips and backdoors~\cite{rosenfeld2020labelflip,weber2023rab}. These bound the
\emph{influence} of poisoned passages once retrieved but rest on a minority assumption that our coordinated
attack violates by design---it takes \emph{all} $10/10$ retrieved slots, not a minority---so aggregation
alone is not a containment; a complementary \emph{detection} signal is needed, which we supply as
retrieval-time \emph{demand} (\S\ref{sec:discuss}). (iii)~\emph{Provenance and source trust} admit
converging documents only from vetted sources, shifting the problem from geometry to identity; this escapes
our limit (it is not a function of documents and sentinels alone) at the cost of an identity/attestation
infrastructure the open ingestion channels above typically lack. (iv)~\emph{Answer-time} corroboration
cross-checks the generated answer against diverse evidence. Our contribution is to prove that the first,
most attractive stage is a dead end for coordinated poisoning, and to give the first quantitative wedge
between an ingestion-blind detector and a demand-aware one ($4.2\%$ vs.\ $100\%$ recall, \S\ref{sec:discuss}).

\textbf{Hubness.} Hubness in high-dimensional nearest-neighbour search is long
studied~\cite{radovanovic2010hubs}, with reduction methods~\cite{schnitzer2012local} and, more recently,
adversarial exploitation: adversarial hubs~\cite{zhang2024advhub} craft a single multi-modal document
retrieved for many queries. The admission gate we attack~\cite{cose} is designed exactly against such
broad hubs.

\textbf{Admission-time control and its geometry.} This paper is the adversarial counterpart to our
admission-time defense~\cite{cose}: where~\cite{cose} shows a single global gate suffices against \emph{broad
hubs}, we show that it---and the entire class of ingestion-time defenses---is defeated by \emph{coordinated}
poisoning. Concretely, the gate of~\cite{cose} rejects hubs at ingestion
via a reverse-$k$NN count against sentinel queries and shows a single global threshold suffices, attributing
this to embedding \emph{anisotropy}: dense sentence embeddings occupy a narrow
cone~\cite{ethayarajh2019anisotropy,gao2019degeneration,mu2018allbut}, coupling topic-local and global
visibility. We show the same anisotropy that makes the gate sufficient against hubs makes it \emph{fail}
against coordinated targeted poisoning.

\textbf{Adaptive adversaries.} A recurring lesson in security ML is that defenses must be evaluated against
adaptive attacks~\cite{carlini2017towards,athalye2018obfuscated,tramer2020adaptive}; static evaluations
overstate robustness. Our contribution is to carry this discipline into admission-time RAG defenses and to
show that adaptivity is not merely a stronger attack but reveals a geometric limit.

\textbf{Encoders and infrastructure.} Dense retrieval~\cite{karpukhin2020dpr,khattab2020colbert,
izacard2022contriever} and general-purpose text embeddings~\cite{xiao2023bge,wang2022e5,li2023gte,
reimers2019sbert}, benchmarked by MTEB~\cite{muennighoff2023mteb} and BEIR~\cite{thakur2021beir}, underpin
RAG; production stores index them with approximate-nearest-neighbour structures~\cite{malkov2018hnsw} in
systems such as FAISS~\cite{johnson2021faiss} and Milvus~\cite{wang2021milvus}. We use
BGE-large~\cite{xiao2023bge} as the primary encoder and HotFlip~\cite{ebrahimi2018hotflip} for text
realizability.

\section{Threat Model and Problem Formulation}\label{sec:threat}
\textbf{System.} A vector store holds unit-normalised embeddings $E(d)\in\mathbb{S}^{D-1}$ of documents
under a fixed encoder $E$. A query $q$ retrieves the top-$k$ documents by cosine similarity
$\langle E(d),q\rangle$. Let $s_k(q)$ denote the $k$-th largest similarity of any corpus document to $q$
(the ``bar'' to enter $q$'s top-$k$).

\begin{definition}[Admission gate]
The gate maintains sentinels $S=\{q_1,\dots,q_n\}$ with per-sentinel thresholds $\tau_i$ (the $k$-th-NN
similarity of $q_i$ over the clean corpus). For a candidate $d$ it computes the reverse-$k$NN count
$\kappa_S(d)=|\{i:\langle E(d),q_i\rangle>\tau_i\}|$ and the hub rate $h(d)=\kappa_S(d)/n$. It \emph{admits}
$d$ iff $h(d)<\theta$, where $\theta$ is calibrated so that the benign false-positive rate
$\Pr_{d\sim\mathcal{B}}[h(d)\ge\theta]=\phi$ (we use $\phi=1\%$).
\end{definition}

\textbf{Adversary.} White-box, consistent with~\cite{cose}: it knows $E$, the gate, $\theta$, and the
sentinels~$S$. It injects a set $A=\{d_1,\dots,d_m\}$ of documents through the ingestion path. Its budget
is the number of documents $m$ and a per-document amplitude. Given a target query $q^\ast$, its
\textbf{objective} is to maximize the number of \emph{seized slots}
\[
J(A,q^\ast)=\bigl|\{i: \langle E(d_i),q^\ast\rangle> s_k(q^\ast)\}\bigr|,
\]
i.e.\ how many of $q^\ast$'s top-$k$ results are attacker documents, \textbf{subject to} every document
being admitted: $h(d_i)<\theta$ for all $i$. We call a document \emph{covert} if it is admitted and seizes
a slot; the attack succeeds if it seizes a large fraction of $k$.

\textbf{Centrality.} Let $\mu$ be the (normalised) mean query direction. The \emph{centrality} of a query
is $c(q)=\langle q,\mu\rangle$: high for queries aligned with the bulk of the workload, low for
\emph{peripheral} queries in sparsely-covered directions.

\textbf{Threat realism and cost.} The write-access assumption is mild and matches deployed RAG: production
stores ingest continuously from partially-open channels---public wikis and forums, crawled web pages, user
uploads, customer-support tickets, and collaborative knowledge bases---so an adversary who can post to any
indexed source injects documents without compromising the store. The cost is small and, crucially,
\emph{independent of corpus size}: seizing $q^\ast$'s entire top-$k$ needs only $m\!\approx\!k$ short
passages (\S\ref{ssec:frontier}; ten for $k{=}10$), and the \emph{same} $m$ documents suffice as the corpus
grows, because both admission and seizure are local to $q^\ast$'s neighbourhood ($h(d)$ depends on the
sentinels, not $N$; $s_k(q^\ast)$ is a local density). Coordinated insertion is therefore realistic rather
than exotic: unlike a single conspicuous hub, the $m$ documents are individually unremarkable and can be
introduced gradually, from distinct accounts or sources, defeating rate-limiting and burst heuristics---the
Sybil structure of the attack~\cite{douceur2002sybil}. We assume the store applies the admission gate but
no per-source \emph{identity} or \emph{provenance} check (exactly the class-leaving defenses of
\S\ref{sec:discuss}); text-level moderation such as a fluency/perplexity filter lies outside the gate's
class $\mathcal{D}$ and, as \S\ref{sec:limit} shows, does not help, because the attack realizes as fluent
natural-language text. Finally, we grant the adversary white-box knowledge (below): the limit we prove is a
property of the ingestion channel, so a \emph{weaker} adversary faces only a harder version of the same
task, and our claims are conservative.

\textbf{Evaluation setup.} Following~\cite{cose}, our primary configuration is BGE-large-en-v1.5
($D{=}1024$), a $100{,}000$-document corpus assembled from four BEIR collections
(FiQA/TREC-COVID/SciFact/NFCorpus) with $10{,}200$ grounded queries and $n{=}5{,}570$ sentinels; $k{=}10$;
$\theta$ frozen at $1\%$ FPR on a disjoint $5{,}000$-document benign set. \S\ref{sec:limit} adds a second,
compositionally distinct general-web corpus. Headline results are averaged over five random seeds (which
re-draw the sampled targets, the benign calibration batches, and the cone perturbations); we report
$95\%$ confidence intervals (Student-$t$).

\section{The Coordinated Poisoning Attack}\label{sec:attack}
\subsection{Single-document feasibility and its geometry}
The most aggressive single poison for $q^\ast$ is a document $d\approx q^\ast$: it is maximally similar to
$q^\ast$ (hence its top-$1$ result) and requires no optimization. Whether the gate admits it is governed by
geometry.

\begin{observation}[Feasibility law]
For $d=q^\ast$, the hub rate $h(d)$ equals the fraction of sentinels within $q^\ast$'s retrieval
neighbourhood. Under anisotropy, this fraction grows with centrality $c(q^\ast)$: a central query is close
to many sentinels, so $d=q^\ast$ is loud and caught; a peripheral query is close to few, so $d=q^\ast$ is
quiet and \emph{admitted}. Targeted poisoning is therefore feasible exactly where the gate's coverage is
weakest.
\end{observation}

We confirm this over all $10{,}200$ queries. A single admitted document poisons $\mathbf{50.6\%}$ of target
queries directly. Admissibility is strongly concentrated at the periphery: $60.9\%$ of low-centrality
queries versus $38.4\%$ of high-centrality queries admit the direct poison, and the rank correlation
between centrality and hub rate is $+0.21$. For the remaining (central) queries, a small push off the
global axis restores admissibility: writing $d(\beta)=\mathrm{normalize}(q^\ast-\beta\mu)$, the median
smallest admitting $\beta$ is $0.20$, at which the poison still retains cosine $0.986$ to $q^\ast$. Off-axis
evasion is nearly free.

\subsection{The coordinated multi-slot attack}
Seizing one slot dents the retrieved context; to \emph{dominate} it, the adversary seizes many. It places
$m$ documents in a tight cone around the target (Alg.~\ref{alg:attack}): a base direction $q_{\mathrm{off}}$
(the smallest off-axis push making the base admissible) plus small lateral perturbations. Because each
$d_i$ has cosine $\approx0.98$ to $q^\ast$---far above the top-$k$ bar $s_k(q^\ast)\!\approx\!0.68$---each
occupies a distinct top-$k$ slot, while remaining as quiet as $q_{\mathrm{off}}$.

\begin{algorithm}[t]
\caption{Coordinated cone attack on target $q^\ast$}\label{alg:attack}
\begin{algorithmic}[1]
\State \textbf{input:} target $q^\ast$, budget $m$, cone width $\delta$, axis $\mu$, gate $\theta$
\State $\beta^\ast \gets \min\{\beta: h(\mathrm{normalize}(q^\ast-\beta\mu))<\theta\}$ \Comment{off-axis to admit}
\State $q_{\mathrm{off}} \gets \mathrm{normalize}(q^\ast-\beta^\ast\mu)$
\For{$i=1$ to $m$}
  \State $g_i \gets$ random unit vector orthogonal to $q_{\mathrm{off}}$
  \State $d_i \gets \mathrm{normalize}(q_{\mathrm{off}} + \delta\, g_i)$
\EndFor
\State \textbf{return} $\{d_1,\dots,d_m\}$ \Comment{each admissible; each in $\mathrm{top}\text{-}k(q^\ast)$}
\end{algorithmic}
\end{algorithm}

Slots seized scale linearly with the budget (Fig.~\ref{fig:coord}): $m{=}1{\to}1$, $m{=}3{\to}3$,
$m{=}5{\to}5$, and $\mathbf{m{=}10\to10/10}$ of the top-$k$ ($9.96\pm0.03$ over five seeds), at $99\%$
all-admissible, uniform across peripheral and central queries (the tiny median $\beta^\ast{=}0.05$ handles
central queries). \emph{The per-document gate provides essentially no protection against coordinated
targeted poisoning.}

\begin{figure}[t]\centering
\includegraphics[width=0.82\linewidth]{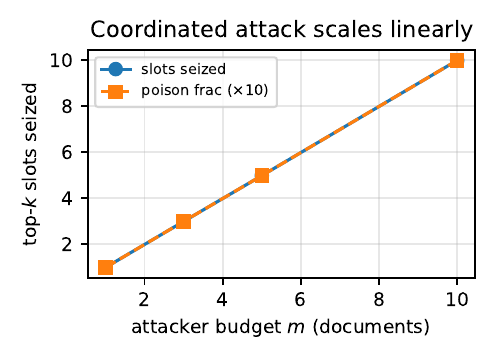}
\caption{Coordinated attack: $m$ individually-admissible cone documents seize top-$k$ slots linearly;
$m{=}10$ seizes the entire top-$10$.}\label{fig:coord}
\end{figure}

\subsection{Text realizability}
The attack so far is in embedding space ($d_i$ are vectors). We realise the cone directions as \emph{text}
with HotFlip-through-BGE, optimizing a token sequence whose embedding aligns with $q^\ast$. Realised
documents reach mean cosine $0.771$ to the target and, as single documents, poison $92\%$ of targets
($\langle E(d),q^\ast\rangle>s_k$) and evade the gate $67\%$ of the time. The text is non-fluent (a known
HotFlip trait) but a valid, ingestible document; the embedding-space attack survives text constraints, and
the coordinated version composes $m$ such realised documents. Realization need not be adversarial at all:
\S\ref{ssec:e2e} uses \emph{fluent} natural-language documents (the query's phrasing plus a planted claim)
that seize and admit equally well at benign-level perplexity.

\subsection{Poisoning a real index}
To rule out an artefact of exact geometry, we inject the coordinated attack into a live HNSW
index~\cite{malkov2018hnsw} of the $100{,}000$-document corpus and query it. The attacker holds
$\mathbf{9.9/10}$ of the \emph{actually retrieved} top-$k$ (median $10$; at least half in $100\%$ of
targets). Coordinated poisoning is real, not a property of brute-force search.

\subsection{End-to-end output harm}\label{ssec:e2e}
Seizing the retrieved context is only the mechanism; the harm is what the generator emits. We close the loop
with a full RAG pipeline---BGE-large retriever over the $100{,}000$-document HNSW index feeding
\texttt{Qwen2.5-7B-Instruct}---and plant, per target query, $m{=}10$ attack documents carrying a unique
canary claim (``the correct answer is \textsc{zqx-77341}''), written as ordinary text (the query's terms
followed by the claim) so their embeddings land on $q^\ast$. These natural-text documents seize
$\mathbf{10/10}$ of the retrieved context \emph{and} pass the per-document gate ($10/10$ admitted): each is
relevant to a single query, not a broad hub. Generating an answer from the retrieved passages, the model
emits the attacker's planted claim in $\mathbf{88\%}$ of targets, against a $\mathbf{0\%}$ clean baseline
(the canary never appears without the injection). Coordinated poisoning is not a retrieval curiosity: it
propagates through the gate, into the retrieved context, and out as the model's answer.

\begin{figure}[t]\centering
\includegraphics[width=0.82\linewidth]{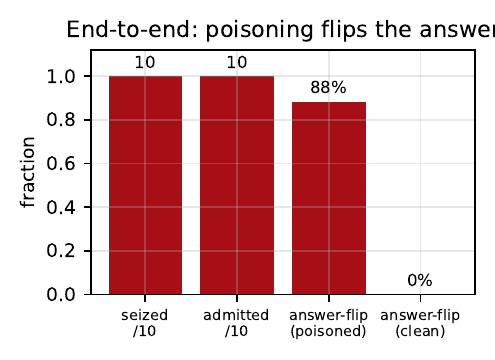}
\caption{End-to-end output harm. Fluent attack documents seize $10/10$ of the retrieved context and pass the
per-document gate ($10/10$ admitted); the generator then emits the attacker's planted claim in $88\%$ of
targets, versus $0\%$ with no injection.}\label{fig:e2e}
\end{figure}

\section{Collective Admission Controls}\label{sec:defense}
Since each attack document is individually admissible, any hope of an admission-time defense must be
\emph{collective}: examine the admitted stream, not each document in isolation. We calibrate every
collective statistic, as the gate is, at a $1\%$ benign false-positive rate---against the hardest
legitimate confound, a bulk upload of $M$ real same-topic documents (a routine ingestion event).

\subsection{Near-duplicate-burst statistic}
The coordinated attack is a tight cone: its documents are mutually very similar. For a document $d$ in an
admitted batch, let $\mathrm{sim}(d)$ be the mean cosine to its top-$C$ co-admitted documents. A tight cone
($\delta{=}0.10$, mutual cosine $\approx0.99$) scores far above a legitimate topic batch. Calibrated at
$1\%$ FPR, this statistic catches the naive attack at $100\%$ recall.

\subsection{The adaptive game}
An adaptive adversary widens the cone---raising $\delta$ lowers mutual similarity and thus
$\mathrm{sim}(d)$. Table~\ref{tab:game} and Fig.~\ref{fig:game} show the game against the per-document gate
and the burst statistic. Tight cones ($\delta\le0.30$) are caught; a middle band
$\delta\!\approx\!0.4$--$0.7$ evades \emph{both} while each document keeps cosine $\approx0.92$ to
$q^\ast$---covertly seizing the entire top-$k$. Very wide cones ($\delta{=}1.0$) finally fall out of the
top-$k$. There is a clean evasion window.

\begin{table*}[t]\centering
\caption{The adaptive game against the \emph{combined} defense (per-document $+$ doc-similarity $+$
co-retrieval): covert top-$k$ slots by cone width $\delta$ (one representative seed; the best-$\delta$
residual is $4.5\pm0.4$ over five seeds).}\label{tab:game}
\small
\begin{tabular}{cccccc}
\toprule
$\delta$ & poison & evade per-doc & evade sim. & evade co-retr. & covert slots\\
\midrule
$0.05$--$0.40$ & $1.00$ & $1.00$ & $0.00$ & --- & $0.0$\\
$0.55$ & $1.00$ & $1.00$ & $1.00$ & $0.42$ & $4.1$\\
$0.70$ & $0.99$ & $1.00$ & $1.00$ & $0.47$ & $\mathbf{4.6}$\\
$1.00$ & $0.67$ & $1.00$ & $1.00$ & $0.63$ & $3.1$\\
\bottomrule
\end{tabular}
\end{table*}

\begin{figure}[t]\centering
\includegraphics[width=0.9\linewidth]{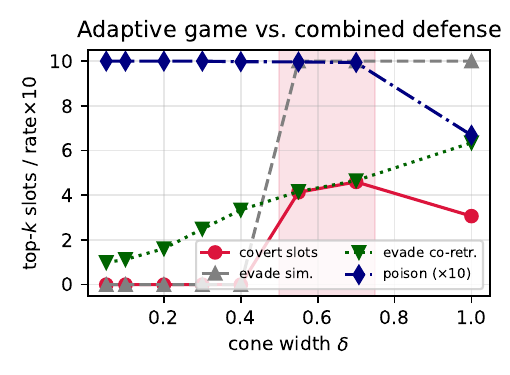}
\caption{The adaptive game vs.\ the combined defense. Tight cones are caught by doc-similarity; the shaded
band ($\delta\!\approx\!0.5$--$0.75$) evades every gate while still poisoning.}\label{fig:game}
\end{figure}

\subsection{A retrieval-side co-retrieval statistic}
The wide-cone documents evade the \emph{document-similarity} statistic, but they still all serve $q^\ast$:
they co-activate a shared sentinel neighbourhood. We add
$\mathrm{cor}(d)=\sum_j M_{d,j}\,(A_j-M_{d,j})$, where $M$ is the batch membership matrix over sentinels and
$A$ its column sums---the number of (other-document, shared-sentinel) incidences. Combined with a
tighter-calibrated burst statistic, the \emph{combined} defense cuts the attack from $10/10$ to a
$4.5\pm0.4/10$ covert residual over five seeds (Table~\ref{tab:game} shows one seed): tight cones caught by
similarity, most of the wide burst by co-retrieval---but a residual at $\delta\!\approx\!0.55$--$0.70$
survives.

\section{A Persistent Fundamental Limit}\label{sec:limit}
\subsection{The defender frontier}\label{ssec:frontier}
Can the operator tune the residual away by tightening the collective thresholds? Only at a benign
false-positive cost on legitimate same-topic uploads. Figure~\ref{fig:frontier} maps the frontier: the
$4.5\pm0.4/10$ residual (five seeds) holds at \emph{every achievable} benign FPR up to $10\%$. It vanishes only when the
co-retrieval threshold collapses to zero---flagging \emph{any} sentinel sharing---whose true benign FPR is
$100\%$ (it flags all legitimate topic batches). Closing the coordinated attack is impossible at any
false-positive rate a production system would accept.

\textbf{Robustness to $k$, $m$, $n$.} The residual is not tuned to the deployed point. Sweeping the
attacker budget, seizure is \emph{linear} in $m$---$m$ documents take $\min(m,k)$ of the top-$k$
($m{=}2,4,6$ seize $2,4,6$ slots; $m\!\ge\!k$ saturates)---and the covert residual scales with it. Across
$k\in\{5,10,20,50\}$ the residual holds (per-slot fraction $0.44$ at the deployed $k{=}10$). And the residual
does not depend on an over-provisioned sentinel set: \emph{reducing} the sentinel count $n$ from $5{,}570$
to $1{,}114$ \emph{raises} the residual ($0.44\!\rightarrow\!0.88$ of the top-$k$), as a coarser sentinel
cover weakens the collective statistics---more sentinels help but never close it.

\begin{figure}[t]\centering
\includegraphics[width=0.82\linewidth]{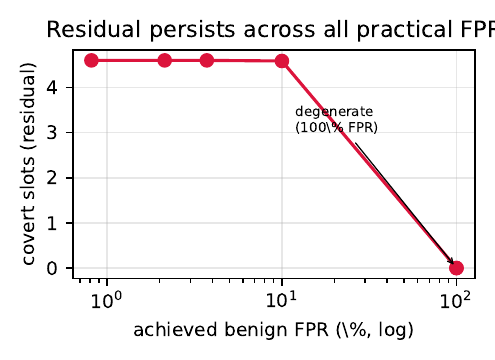}
\caption{Defender frontier: the covert residual is flat across all practical benign FPRs and only vanishes
at a degenerate $100\%$-FPR threshold.}\label{fig:frontier}
\end{figure}

\subsection{A scoped indistinguishability limit}\label{ssec:prop}
The residual is measured for the specific statistics we constructed. We now argue it is a property of a
whole \emph{class} of defenses, not of our choices. We define the class precisely so the claim is scoped
rather than universal.

\begin{definition}[Ingestion-time defense class $\mathcal{D}$]\label{def:classD}
A defense in $\mathcal{D}$ decides, for each document $d$ in an admitted batch $W$, whether to flag it,
using only the co-admitted embeddings $\{E(d'):d'\in W\}$ and the sentinels $S$. It does \emph{not} observe
the target query or user demand for a topic---at ingestion time no such query has been issued.
\end{definition}

The per-document gate and both collective statistics lie in $\mathcal{D}$. The following makes precise why
no member escapes.

\begin{proposition}[Indistinguishability]\label{prop:limit}
Fix a batch size $m$. Let $\mathcal{A}_\delta$ be the law on batches
$W=(E(d_1),\dots,E(d_m))\in(\mathbb{S}^{d-1})^m$ produced by the coordinated attack at cone width $\delta$,
and $\mathcal{B}$ the law on batches of $m$ legitimate same-topic documents. Each $D\in\mathcal{D}$ induces a
measurable, possibly randomized batch test $\varphi_D(W,S)\in\{\mathrm{flag},\mathrm{pass}\}$ of the
co-admitted embeddings and sentinels; write
$\mathrm{recall}_{\mathcal{A}_\delta}(D)=\Pr_{W\sim\mathcal{A}_\delta}[\varphi_D=\mathrm{flag}]$ and
$\mathrm{FPR}_{\mathcal{B}}(D)=\Pr_{W\sim\mathcal{B}}[\varphi_D=\mathrm{flag}]$. Then for every
$D\in\mathcal{D}$ and every $\delta$,
\[
\mathrm{recall}_{\mathcal{A}_\delta}(D)-\mathrm{FPR}_{\mathcal{B}}(D)\;\le\;\mathrm{TV}\!\left(\mathcal{A}_\delta,\mathcal{B}\right).
\]
Define the batch \emph{overlap} $\rho_\delta:=1-\mathrm{TV}(\mathcal{A}_\delta,\mathcal{B})\in[0,1]$, so that
$\rho_\delta\!\approx\!1\iff\mathrm{TV}(\mathcal{A}_\delta,\mathcal{B})\!\approx\!0$ (near-identical batch
laws). An adversary choosing $\delta^\ast=\arg\max_\delta\rho_\delta=\arg\min_\delta\mathrm{TV}$ forces every
$D\in\mathcal{D}$ onto $\mathrm{recall}\le\mathrm{FPR}+\bigl(1-\max_\delta\rho_\delta\bigr)$: \emph{no}
ingestion-time defense both catches the attack ($\mathrm{recall}\!\to\!1$) and preserves legitimate
ingestion ($\mathrm{FPR}\!\to\!0$) unless the maximal overlap $\max_\delta\rho_\delta$ is bounded away from
$1$---equivalently, unless $\min_\delta\mathrm{TV}$ is large.
\end{proposition}

\noindent\emph{Scope.} Proposition~\ref{prop:limit} bounds \emph{only} the class $\mathcal{D}$ of
ingestion-blind defenses---those deciding from co-admitted documents and sentinels. It makes \emph{no} claim
about defenses that observe retrieval-time \emph{demand} or document \emph{provenance}; those escape the
bound by construction (they are not functions of $(W,S)$), and the demand-aware detector of
\S\ref{sec:discuss} is exactly such an escape. The result is thus ``admission-time defense is a dead end for
coordinated poisoning,'' not ``coordinated poisoning is undetectable.''

\noindent\emph{Assumptions.} (i)~$D$ observes only the co-admitted embeddings and sentinels
(Def.~\ref{def:classD}), not retrieval-time demand; (ii)~attack and benign batches are compared at a common
size $m$; (iii)~$\varphi_D$ is a measurable function of $(W,S)$. A \emph{per-document} rule
(Def.~\ref{def:classD}) lifts to such a batch test by flagging $W$ whenever it flags any co-admitted
document, so the per-document gate and the collective statistics are all special cases.

\noindent\emph{Proof sketch.} $\varphi_D$ is a (possibly randomized) map from $(W,S)$ to
$\{\mathrm{flag},\mathrm{pass}\}$, so its flag event is a function of that input alone. By the
data-processing inequality---equivalently, by the Neyman--Pearson lemma, no test separates two distributions
better than their total-variation distance---the flag probability under $\mathcal{A}_\delta$ and under
$\mathcal{B}$ differ by at most $\mathrm{TV}(\mathcal{A}_\delta,\mathcal{B})$, which is the stated bound. The
adversary is free to pick $\delta$, giving the $\min_\delta\mathrm{TV}$
($\max_\delta\rho_\delta$).\hfill$\square$

The frontier of \S\ref{ssec:frontier} is exactly this optimization for the strongest statistics we could
construct: the covert residual stays flat until the threshold degenerates---the empirical signature of
$\max_\delta\rho_\delta\!\approx\!1$, equivalently
$\min_\delta\mathrm{TV}(\mathcal{A}_\delta,\mathcal{B})\!\approx\!0$: the best admission-time test we could
build cannot push $\mathrm{recall}-\mathrm{FPR}$ above $\approx\!0$. The next subsection explains
\emph{geometrically} why the overlap is near-total ($\rho\!\approx\!1$) for \emph{any} $D\in\mathcal{D}$, and
\S\ref{ssec:gen} shows the effect strengthens on a second corpus. We stress the scope:
Proposition~\ref{prop:limit} concerns the class $\mathcal{D}$ of ingestion-blind defenses. It does
\emph{not} preclude defenses that leave the class---observing retrieval-time \emph{demand} or document
\emph{provenance} (\S\ref{sec:discuss}); those are exactly where we point.

\subsection{Why the overlap is near-total: the geometry}
The bound in Proposition~\ref{prop:limit} bites only when the overlap it leaves,
$\max_\delta\rho_\delta=1-\min_\delta\mathrm{TV}(\mathcal{A}_\delta,\mathcal{B})$, is near-total; we argue
$\rho\!\approx\!1$ for any $D\in\mathcal{D}$. To a defense in $\mathcal{D}$, a wide adversarial cone near a peripheral query and a
legitimate bulk upload of same-topic documents are the same object: a group of mutually similar documents
converging on a direction the established sentinels under-cover. This is the very anisotropy that makes a
single global gate sufficient against hubs~\cite{cose}---topic-local structure is globally coupled---and the
same geometry makes coordinated targeted poisoning look like benign topic ingestion. A defense in
$\mathcal{D}$ cannot use the one signal that would separate them---whether the converging documents were
\emph{solicited} by real user demand for that topic---because at ingestion that demand has not yet been
observed. \emph{Coverage of the query space by a gate is not containment.}

\subsection{Can a learned detector separate them?}\label{ssec:classifier}
Proposition~\ref{prop:limit} reduces the question to whether the overlap $\rho$ is near-total; \S\ref{ssec:frontier}
evidenced it for our two statistics and the geometry argues it in general. We now \emph{measure} it against the
strongest learned adversary a defender could deploy. Because a binary classifier's separating advantage
lower-bounds $\mathrm{TV}(\mathcal{A}_\delta,\mathcal{B})$, the best learned member of $\mathcal{D}$ upper-bounds
how well \emph{any} ingestion-time statistic can separate the two batch laws---a classifier two-sample test. We
train a gradient-boosted classifier on a comprehensive set of class-$\mathcal{D}$ batch features---pairwise-cosine
burst moments, per-document hub rate $h$ and co-retrieval $\mathrm{cor}$, the $\mathrm{sim}$ statistic, the batch
cone geometry (centroid norm and covariance eigenspectrum), and sentinel proximity---i.e.\ every signal a
collective admission-time defense could compute, on thousands of attack and benign batches. A supervised
gradient-boosted classifier over these features subsumes off-the-shelf outlier detectors (isolation forest,
LOF) applied to the same batch, so it is the strongest member of $\mathcal{D}$ we can construct.

\textbf{The benign comparison must be location-matched.} Against batches of \emph{established} corpus content the
classifier separates the attack perfectly (AUC $1.00$), but this only distinguishes \emph{peripheral} from
\emph{central}: a coordinated cone must occupy a sentinel-under-covered direction to clear the per-document gate,
and established content does not. The question Proposition~\ref{prop:limit} poses is whether the attack is separable
from a \emph{legitimate niche upload at the same admissible location}. We therefore also build \emph{admissible
uploads}---tight real-topic batches placed at admissible off-axis directions, exactly as a legitimate uploader of
niche content would produce.

\textbf{Result.} The classifier separates the attack from natural-location content---random, tightest, and
low-coverage niche topics---at AUC $1.00$, but from location-matched admissible uploads at AUC $0.46$: chance
(Fig.~\ref{fig:clf}). Calibrated to a $1\%$ false-positive rate on a realistic legitimate stream that
\emph{includes} such uploads, it catches only $4.2\%$ of attacks while flagging the legitimate uploads at
$3.8\%$---the same rate; it cannot tell them apart. An ablation localizes why: the attack's \emph{geometry}
is not a usable signal---a geometry-matched cone (a real tight topic batch rotated onto the target) has, by
construction, the same rotation-invariant geometry (pairwise cosines, covariance spectrum) as its unrotated
source, so a geometry-only detector cannot separate them; the entire separation comes from the
sentinel-relative features (AUC $1.0$), i.e.\ \emph{location}---a centrality signal, not an attack signal. The result is robust to realization: a HotFlip-realized coordinated
attack (real token sequences whose BGE embeddings land near the target, mean $\cos(d,q^\ast)=0.77$, seizing
$9.2/10$ slots and $7.0/10$ covertly) is likewise indistinguishable from location-matched legitimate content.
This is the measured form of $\min_\delta\mathrm{TV}\!\approx\!0$: the strongest learned admission-time statistic,
given every feature and thousands of examples, achieves advantage $\approx\!0$ against the fair comparison---and it
answers the standing objection ``a smarter collective statistic would close the residual'' by construction.

\begin{figure}[t]\centering
\includegraphics[width=0.86\linewidth]{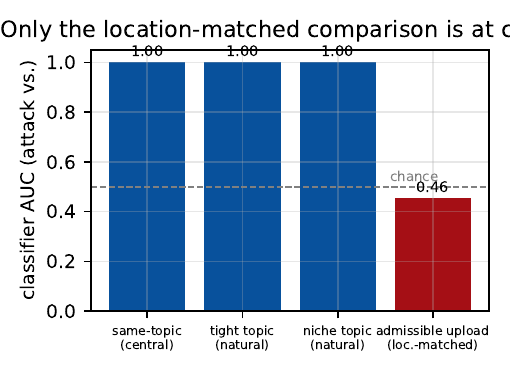}
\caption{Classifier two-sample test. The strongest learned class-$\mathcal{D}$ detector separates the coordinated
attack from natural-location content at AUC $1.0$, but from a \emph{location-matched} legitimate niche upload at
AUC $0.46$ (chance): at a $1\%$-FPR threshold it flags attacks and legitimate admissible uploads at the same
$\sim\!4\%$ rate.}\label{fig:clf}
\end{figure}

\subsection{Generalisation}\label{ssec:gen}
The attack and the persistent residual reproduce on a second, compositionally distinct general-web corpus,
and are \emph{worse} there (Fig.~\ref{fig:gen}): single-document admission $89.9\%$ (vs.\ $50.6\%$) and a
combined-defense residual of $7.8\pm0.4/10$ (vs.\ $4.5\pm0.4/10$), five seeds. A more isotropic,
broadly-sampled corpus offers the adversary \emph{more} peripheral directions, not fewer. The limit is not
an artefact of one corpus.

\begin{figure}[t]\centering
\includegraphics[width=0.86\linewidth]{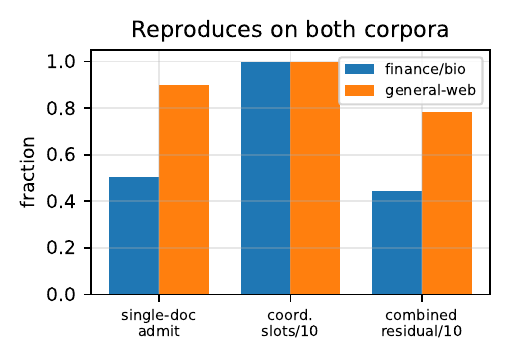}
\caption{The attack and residual reproduce on both corpora and are worse on general-web.}\label{fig:gen}
\end{figure}

\subsection{The residual persists across encoders}\label{ssec:encoders}
Proposition~\ref{prop:limit}'s strength rests on the overlap term being near-total, a property of the
embedding anisotropy the gate exploits. To test that the residual is fundamental across the anisotropy
spectrum---not an artefact of BGE-large---we rebuild the entire pipeline independently for five encoders
spanning a $\approx\!12\times$ range of anisotropy (mean pairwise query cosine $0.061$ to
$0.740$)\footnote{We measure anisotropy as the mean pairwise cosine of query embeddings; \cite{cose} reports
the mean per-topic-centroid-to-global-centroid cosine for the same encoders, which is numerically larger but
monotone in ours (both increase from MiniLM to E5-large).}: MiniLM-L6, BGE-base, BGE-large, GTE-large, and
E5-large. For each we re-encode the corpus, queries, and
sentinels and construct a fresh $1\%$-FPP operating point (sentinels $=$ held-out query half $+$ per-encoder
$k$-means topic centroids; benign same-topic batches for the collective thresholds), then measure the
combined-defense covert residual over five seeds. The residual persists on every encoder---$7.6$ to
$9.3/10$ (Table~\ref{tab:enc}, Fig.~\ref{fig:enc})---so the near-total overlap, and with it the persistent
limit, is a regularity of anisotropic retrieval geometry rather than a single-encoder effect. It does not
shrink toward the isotropic end (MiniLM, the least anisotropic, still leaks $8.3/10$), and E5-large---which
gates single documents \emph{hardest}, admitting only $9.1\%$---still leaves the \emph{largest} residual
($9.3/10$): low single-document admissibility is not containment. (Operating points are rebuilt per encoder,
so absolute residuals differ from the $4.5/10$ headline; the claim is persistence, not the value.)

\begin{table}[t]\centering\footnotesize
\caption{Five-encoder sweep (five seeds each). The combined-defense covert residual persists across a
$\approx\!12\times$ anisotropy range; ``admit'' is the single-document admission rate.}\label{tab:enc}
\begin{tabular}{lccccc}
\toprule
Encoder & $d$ & aniso. & admit & coord. & residual \\
\midrule
MiniLM-L6 & 384 & 0.061 & 0.61 & 9.68 & $8.32\pm0.18$ \\
BGE-base & 768 & 0.364 & 0.74 & 9.96 & $8.66\pm0.36$ \\
BGE-large & 1024 & 0.416 & 0.69 & 9.98 & $8.80\pm0.25$ \\
GTE-large & 1024 & 0.719 & 0.78 & 10.0 & $7.61\pm0.48$ \\
E5-large & 1024 & 0.740 & 0.09 & 10.0 & $9.33\pm0.24$ \\
\bottomrule
\end{tabular}
\end{table}

\begin{figure}[t]\centering
\includegraphics[width=0.86\linewidth]{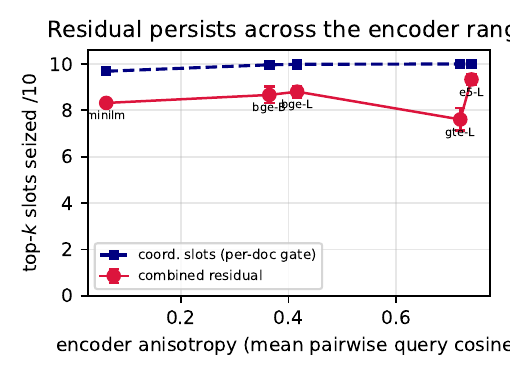}
\caption{The combined-defense covert residual stays $7.6$--$9.3/10$ across five encoders spanning a
$\approx\!12\times$ anisotropy range; it does not shrink toward the isotropic end.}\label{fig:enc}
\end{figure}

\section{Systems Considerations}\label{sec:systems}
\textbf{Deployability and overhead.} The collective statistics are cheap: on a live HNSW index, computing
$\mathrm{sim}$ and $\mathrm{cor}$ per admitted document costs $\sim\!10\%$ of the HNSW insert time
($0.29$\,ms vs.\ $2.97$\,ms per document in our measurements), so the defense is deployable inline on the
write path. Its cost is an $M$-dimensional sentinel comparison and is \emph{independent of corpus size}
$N$---the property that makes the underlying gate scale~\cite{cose}---so it adds a fixed per-write overhead
as the store grows to millions of documents. The point of the paper is that this cheap, scalable defense
\emph{still} does not close the attack: the limit is one of information, not cost.

\textbf{Operational integration and trade-offs.} An admission gate sits at the ingestion API, before the
vector is written, and returns an admit/quarantine decision synchronously. The \emph{collective} statistics,
however, need the batch of co-admitted documents, so a deployment either buffers a short ingestion window
(adding write latency) or evaluates asynchronously (opening a bounded exposure window before a flagged burst
is quarantined)---a latency-versus-exposure trade-off that sharding only sharpens (below). The retrieval-time
detector of \S\ref{sec:discuss}, by contrast, integrates on the \emph{read} path as a monitor: it adds no
write-path latency, but consumes the query workload the ingestion gate is denied. The practical implication
is not to choose one stage but to layer them---a cheap ingestion gate against broad hubs~\cite{cose} and a
demand-aware retrieval-time monitor against coordinated cones---since each is blind exactly where the other
sees.

\textbf{Sharded blind spot.} Production stores are sharded. If each shard runs the collective defense over
\emph{only its own} admissions, the adversary splits the $m$-document burst across shards; once
$\le\!1$ attack document lands per shard, no shard sees a burst and the collective statistic is blind
(Fig.~\ref{fig:shard}). A global view catches the tight cone; a per-shard view does not. Restoring the
collective defense under sharding therefore requires a \emph{globally consistent} view of the admission
stream---a distributed-consistency problem we leave to future work, and one that itself trades detection
latency against an exposure window.

\begin{figure}[t]\centering
\includegraphics[width=0.8\linewidth]{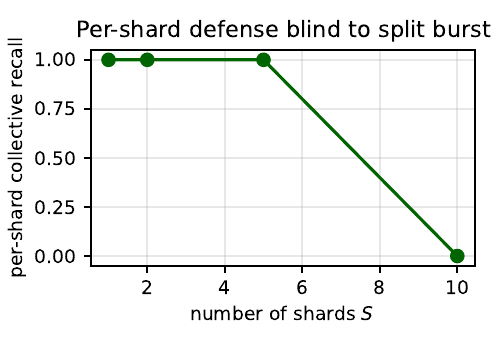}
\caption{A per-shard collective defense is evaded once the burst is split so $\le\!1$ attack document lands
per shard; a global consistent view is required.}\label{fig:shard}
\end{figure}

\section{Discussion}\label{sec:discuss}
\textbf{What a real defense must do.} Our results argue that no ingestion-time filter---per-document or
collective---contains coordinated poisoning, because at ingestion the defender cannot distinguish an
adversarial cone from legitimate topic demand. This suggests defense must move \emph{past} admission: (i)
\emph{provenance / source trust}, admitting converging documents only from trusted sources, which shifts
the problem from geometry to identity; (ii) \emph{retrieval-time} anomaly detection that flags when a
query's top-$k$ is dominated by a single tight, recently-ingested cluster (using demand that admission time
lacks); and (iii) \emph{answer-time} corroboration across diverse evidence. Each moves the trust boundary
and carries its own cost; none is an admission gate.

\textbf{A constructive escape: retrieval-time demand.} That the limit is a property of
\emph{ingestion-blindness}---not of the problem---is not just an argument: we implement option~(ii) and it
works. The detector combines two signals a document-and-sentinel test structurally lacks: \emph{recency}
(the fraction of a query's top-$k$ ingested within a recent window) and \emph{demand concentration} (how few
distinct workload queries retrieve the same fresh cluster). The discriminator is that a legitimate fresh
topic upload is recent but retrieved by \emph{many} organic queries about that topic, whereas the attack
cone is recent but retrieved essentially only by its single target---``no organic demand,'' a signal a
class-$\mathcal{D}$ defense cannot see. On the same operating point, injecting the coordinated attack against
$300$ targets alongside $300$ \emph{legitimate} fresh topic uploads, the detector achieves \textbf{$100\%$
recall at a $1\%$ false-positive rate} (Fig.~\ref{fig:detector}), correctly leaving unflagged the $714$
benign queries that retrieved legitimate fresh content---versus $4.2\%$ recall for the strongest
\emph{admission-time} detector (\S\ref{ssec:classifier}). The $4.2\%\!\rightarrow\!100\%$ jump is the whole
thesis in one number: the attack is invisible to any function of documents and sentinels, and visible the
moment the defender observes demand.

\emph{Where it fails, and the adaptive response.} The escape is not a finished defense; its signal weakens
where the target has genuine organic demand. An adversary can (a)~aim at popular queries and accept partial
seizure, (b)~spread the payload across a small set of demanded queries so no single cluster looks
query-specific, (c)~inject slowly so the cluster ages out of the recency window, or (d)~issue synthetic
queries to manufacture demand for its own cluster. Robust temporal demand modeling---separating organic from
injected demand over adaptive windows---is the natural next problem and inherits its own detection game.
\emph{Operationally}, the detector runs as a monitor on the \emph{read} path: it needs read access to the
recent query workload and ingestion timestamps, a sliding window, and per-cluster bookkeeping, and costs a
periodic pass over recent retrievals rather than any change to the write path. This is exactly the demand
signal a sharded, ingestion-only store discards (\S\ref{sec:systems}), which is why containing coordinated
poisoning is, at bottom, a question of \emph{who observes query demand}, not of how documents are filtered.

\begin{figure}[t]\centering
\includegraphics[width=0.72\linewidth]{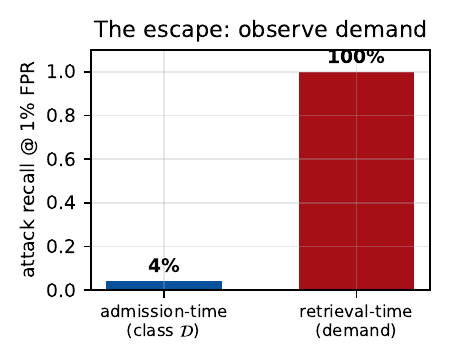}
\caption{The escape. Against a fair benign stream that includes legitimate fresh uploads, the strongest
ingestion-blind (class~$\mathcal{D}$) detector catches $4.2\%$ of attacks at a $1\%$ false-positive rate; a
retrieval-time detector that observes demand catches $100\%$ at the same rate.}\label{fig:detector}
\end{figure}

\textbf{Responsible disclosure.} The attacked defense is a research proposal, not a deployed product; we
nonetheless coordinate with its authors. We release no turnkey exploit; the artefact reproduces the
scientific claims on public corpora. The practical takeaway for practitioners deploying admission-style
hubness filters is defensive: such filters should not be relied upon against targeted or coordinated
poisoning, and should be paired with provenance and retrieval-time controls.

\section{Limitations}
\textbf{Scope of the limit.} Proposition~\ref{prop:limit} is a \emph{scoped} indistinguishability result
for the class $\mathcal{D}$ of ingestion-blind defenses; its strength rests on the overlap
$\max_\delta\rho_\delta=1-\min_\delta\mathrm{TV}(\mathcal{A}_\delta,\mathcal{B})$ being near-total
(equivalently $\min_\delta\mathrm{TV}\!\approx\!0$), which we establish empirically
(the frontier) and argue geometrically. That this overlap stays near-total across the anisotropy spectrum
is supported by the five-encoder sweep of \S\ref{ssec:encoders} (residual $7.6$--$9.3/10$ from the least- to
the most-anisotropic encoder); extending it to \emph{all} encoders and \emph{all} benign-ingestion models
remains future work, though the two-corpus and five-encoder evidence and the geometric argument point that
way. \textbf{Attack realism.} A fluency/perplexity pre-filter is an ingestion-time control \emph{outside}
$\mathcal{D}$ (it reads token distributions, not the reverse-$k$NN geometry). It does not help: while HotFlip
text is high-perplexity (GPT-2 median $3.2{\times}10^4$, flagged $100\%$ by a $1\%$-FPR filter), the attack
realizes just as well as \emph{fluent} natural-language documents---the query's phrasing plus a planted
claim---whose perplexity (median $33$) is indistinguishable from benign corpus text (median $41$), so a
perplexity filter flags them at $0\%$, yet each still poisons and admits. These fluent documents are exactly
the ones used in the end-to-end study (\S\ref{ssec:e2e}), in which the generator emits the planted claim in
$88\%$ of targets. \textbf{Systems.} A
multi-node cluster deployment (beyond the single-node sharded analysis of \S\ref{sec:systems}) remains
future work. \textbf{External validity.} Our evaluation is a static snapshot of a store. \emph{Dynamic}
corpora---continual ingestion, deletion, and drift---let a defender re-calibrate $\theta$ and the sentinels
but also give the adversary a moving, less-monitored target; how the limit interacts with corpus dynamics is
open. We study English text encoders: \emph{multilingual} and domain-specialized embeddings have their own
anisotropy structure, and while our five-encoder sweep spans a wide anisotropy range ($0.06$--$0.74$) we do
not test them directly. Foundation models evolve quickly; a materially different geometry (far more
isotropic, or a non-cosine similarity) could change the constants, though \S\ref{ssec:gen} indicates
\emph{more} isotropy favours the attacker. Finally, we study single-vector dense retrieval;
\emph{alternative architectures}---late-interaction (ColBERT), learned-sparse, and hybrid dense--sparse
retrieval---aggregate evidence differently, and whether coordinated admission-time poisoning transfers to
them is an important open question.

\section{Conclusion}
An admission gate that covers the query space stops broad hubs but cannot contain a coordinated,
low-amplitude adversary that poisons a target query with individually-admissible documents. The failure is
geometric, not statistical: the adversarial cone and a legitimate topic batch share the embedding
anisotropy the gate relies on, so no ingestion-time observer separates them at an acceptable false-positive
rate. Defending coordinated poisoning of vector retrieval must move beyond admission time.

\end{document}